\documentclass[letterpaper, 10 pt, conference]{ieeeconf}  

\IEEEoverridecommandlockouts                              

\usepackage{graphicx} 
\usepackage{amsmath} 
\usepackage{amssymb}  

\usepackage{bm}
\usepackage{color}
\usepackage{tabularx}
\usepackage{booktabs}
\usepackage{multirow}
\usepackage{mathtools}
\usepackage{algorithmic}
\usepackage{algorithm}
\usepackage{savesym}
\usepackage{units}
\savesymbol{AND}
\savesymbol{OR}
\savesymbol{NOT}
\savesymbol{TO}
\savesymbol{COMMENT}
\savesymbol{BODY}
\savesymbol{IF}
\savesymbol{ELSE}
\savesymbol{ELSIF}
\savesymbol{FOR}
\savesymbol{WHILE}

\newcommand{\norm}[1]{\left\lVert#1\right\rVert}

\newcommand{\diag}{\mathrm{diag}}

\newtheorem{theorem}{Theorem}

\newtheorem{remark}{Remark}

\newtheorem{lemma}{Lemma}
\newtheorem{corollary}{Corollary}

\makeatletter
\newcommand{\eqnum}{\refstepcounter{equation}\textup{\tagform@{\theequation}}}
\makeatother

\def\tube-#1{tube\nobreakdash-#1}
\def\df-#1{df\nobreakdash-#1}

\title{\LARGE \bf
A System Level Approach to Tube-based Model Predictive Control
}

\author{Jerome Sieber$^{1}$, Samir Bennani$^{2}$, and Melanie N. Zeilinger$^{1}$
\thanks{*This work was supported by the European Space Agency (ESA) under NPI 621-2018 and the Swiss Space Center (SSC).%
}
\thanks{$^{1}$J. Sieber and M. N. Zeilinger are members of the Institute for Dynamic Systems and Control (IDSC), ETH Zurich, 8092 Zurich, Switzerland
        {\tt\small \{jsieber,mzeilinger\}@ethz.ch}%
}
\thanks{$^{2}$S. Bennani is a member of ESA-ESTEC, Noordwijk 2201 AZ, The Netherlands
        {\tt\small samir.bennani@esa.int}%
}
}

\begin{document}

\maketitle
\thispagestyle{empty}
\pagestyle{empty}

\begin{abstract}
Robust tube-based model predictive control (MPC) methods address constraint satisfaction by leveraging an a~priori determined tube controller in the prediction to tighten the constraints. This paper presents a system level \tube-MPC~(SLTMPC) method derived from the system level parameterization (SLP), which allows optimization over the tube controller online when solving the MPC problem, which can significantly reduce conservativeness. We derive the SLTMPC method by establishing an equivalence relation between a class of robust MPC methods and the SLP. Finally, we show that the SLTMPC formulation naturally arises from an extended SLP formulation and show its merits in a numerical example.
\end{abstract}


\section{INTRODUCTION}
The availability of powerful computational hardware for embedded systems together with advances in optimization software has established model predictive control (MPC) as the principal control method for constrained dynamical systems. MPC relies on a sufficiently accurate system model to predict the system behavior over a finite horizon. In order to address constraint satisfaction in the presence of bounded uncertainties, robust MPC methods~\cite{Bemporad1999} typically optimize over a feedback policy and tighten constraints in the prediction problem. The two most common robust MPC methodologies are disturbance feedback MPC (\df-MPC)~\cite{Lofberg2003,Goulart2006} and tube-based MPC~\cite{Chisci2001,Langson2004,Mayne2005}, which mainly differ in their parameterization of the feedback policy. Disturbance feedback parametrizes the input in terms of the disturbance, while tube-based methods split the system dynamics into nominal and error dynamics and parametrize the input in terms of those two quantities. In general, \df-MPC is computationally more demanding, yet less conservative than tube-based methods. In this paper, we present a tube-based MPC method, which lies at the intersection of these two extremes by offering a trade-off between computational complexity and conservativeness.

We derive the proposed method by leveraging the system level parameterization (SLP), which has recently received increased attention. It was introduced as part of the system level synthesis (SLS) framework~\cite{Anderson2019}, in particular in the context of distributed optimal control~\cite{Wang2018a}. The SLP offers the advantage that convexity is preserved under any additional convex constraint imposed on the parameterization and enables optimization over closed-loop trajectories instead of controller gains~\cite{Anderson2019}. We relate the SLP to \df-MPC and exploit this relation in order to formulate a system level \tube-MPC~(SLTMPC) method.

\textit{Related Work:} A tube-based MPC method reducing conservativeness by using multiple pre-computed tube controllers was introduced in~\cite{Kogel2020}. In contrast, the proposed SLTMPC method directly optimizes over the controller gains online. At the intersection of SLP and MPC, the first SLP-based MPC formulation was introduced in~\cite{Anderson2019} and then formalized as SLS-MPC in~\cite{Chen2020}, where both additive disturbances and parametric model uncertainties were integrated in the robust MPC problem. A more conservative variant of this approach was already introduced in the context of the linear quadratic regulator (LQR) in~\cite{Dean2018}. In a distributed setting, an SLP-based MPC formulation was proposed in~\cite{Alonso2019}, which builds on the principles of distributed SLS presented in~\cite{Wang2018a} and was later extended to distributed explicit SLS-MPC~\cite{Alonso2020} and layered SLS-MPC~\cite{Li2020}.

\textit{Contributions:} We consider linear time-invariant dynamical systems with additive uncertainties. In this setting, we first show the equivalence of \df-MPC and the SLP, before using this relation to analyze the inherent tube structure present in the SLP. Based on this analysis, we propose a SLTMPC formulation, which allows optimization over the tube controller in the online optimization problem and thus reduces conservativeness compared to other tube-based methods. Additionally, we show that the SLTMPC can be derived from an extended version of the SLP and outline extensions to distributed and explicit SLTMPC formulations. Finally, we show the effectiveness of our formulation on a numerical example.

The remainder of the paper is organized as follows: Section~\ref{sec:preliminaries} introduces the notation, the problem formulation, and relevant concepts for this paper. We present the equivalence relation between the SLP and \df-MPC in Section~\ref{sec:SL-perspective}, before deriving the SLTMPC formulation in Section~\ref{sec:SLTMPC}. Finally, Section~\ref{sec:numerical_section} presents a numerical application of SLTMPC and Section~\ref{sec:conclusions} concludes the paper.

\section{PRELIMINARIES}\label{sec:preliminaries}
\subsection{Notation}
In the context of matrices and vectors, the superscript~$^{r,c}$ denotes the element indexed by the r-th row and c-th column. Additionally, $^{:r,c:}$ refers to all rows from 1 to $r-1$ and all columns from $c$ to the end, respectively. The notation $\mathcal{S}^i$ refers to the i-th Cartesian product of the set $\mathcal{S}$, i.e. $\mathcal{S}^i = \mathcal{S} \times \dots \times \mathcal{S} = \left\lbrace (s_0, \dots, s_{i-1}) \mid s_j \in \mathcal{S} \; \forall j = 0, \dots, i-1\right\rbrace$.

\subsection{Problem Formulation}
We consider linear time-invariant (LTI) dynamical systems with additive disturbances of the form
\begin{equation}\label{eq:dynamics}
x_{k+1} = A x_k + B u_k + w_k,
\end{equation}
with $w_k \in \mathcal{W}$, where $\mathcal{W}$ is a compact set. Here, we assume time-invariant polytopic sets $\mathcal{W} = \{ w \in \mathbb{R}^n \mid S w \leq s \}$, but the results can easily be generalized to other compact convex sets. The system~\eqref{eq:dynamics} is subject to polytopic state and input constraints containing the origin in their interior
\begin{equation}\label{eq:constraints}
\mathcal{X} \!=\! \{ x \!\in\! \mathbb{R}^n \!\mid\! H_{x} x \leq h_{x}\}, \ \mathcal{U} \!=\! \{ u \!\in\! \mathbb{R}^m \!\mid\! H_{u} u \leq h_{u}\}.
\end{equation}
This paper presents a system level \tube-MPC formulation, i.e. a tube-based MPC formulation derived via a system level parameterization, which enables online optimization of the tube control law. To formulate this method, we first show the equivalence of the system level parameterization and disturbance feedback MPC, both of which are introduced in the following sections.

\subsection{System Level Parameterization (SLP)}
We consider the finite horizon version of the SLP, as proposed in~\cite{Anderson2019}. The main idea of the SLP is to parameterize the controller synthesis by the closed-loop system behavior, which restates optimal control as an optimization over closed-loop trajectories and renders the synthesis problem convex under convex constraints on the control law allowing, e.g., to impose a distributed computation structure~\cite{Wang2018a}. We define $\mathbf{x}$, $\mathbf{u}$, and $\mathbf{w}$ as the concatenated states, inputs, and disturbances over the horizon~$N$, respectively, and $\bm{\delta}$ as the initial state $x_0$ concatenated with the disturbance sequence $\mathbf{w}$: ${\mathbf{x} = [ x_0, x_1, \dots, x_N ]^\top}$, $\mathbf{u} = [ u_0, u_1, \dots, u_N ]^\top$, $\bm{\delta} = [ x_0, \mathbf{w} ]^\top$. The dynamics can then be compactly defined as trajectories over the horizon $N$,
\begin{equation}\label{eq:stacked_dynamics}
\mathbf{x} = \mathcal{ZA}\mathbf{x} + \mathcal{ZB}\mathbf{u} + \bm{\delta},
\end{equation}
where $\mathcal{Z}$ is the down-shift operator and $\mathcal{ZA}$, $\mathcal{ZB}$ are the corresponding dynamic matrices
\begin{equation*}
\mathcal{ZA} = \begin{bmatrix} 0 & \hspace{-.3em}\dots & \hspace{-.3em}\dots & \hspace{-.3em}0 \\ A  & \hspace{-.3em}0 & \hspace{-.3em}\dots & \hspace{-.3em}0 \\ \vdots & \hspace{-.3em}\ddots & \hspace{-.3em}\ddots & \hspace{-.3em}\vdots \\ 0 & \hspace{-.3em}\dots & \hspace{-.3em}A & \hspace{-.3em}0 \end{bmatrix}, \,
\mathcal{ZB} = \begin{bmatrix} 0 & \hspace{-.3em}\dots & \hspace{-.3em}\dots & \hspace{-.3em}0 \\ B & \hspace{-.3em}0 & \hspace{-.3em}\dots & \hspace{-.3em}0 \\ \vdots & \hspace{-.3em}\ddots & \hspace{-.3em}\ddots & \hspace{-.3em}\vdots \\ 0 & \hspace{-.3em}\dots & \hspace{-.3em}B & \hspace{-.3em}0 \end{bmatrix}.
\end{equation*}
In order to formulate the closed-loop dynamics, we define an LTV state feedback controller $\mathbf{u} = \mathbf{Kx}$, with
\begin{equation*}
\mathbf{K} = \begin{bmatrix} K^{0,0} & 0 & \dots & 0 \\ K^{1,1} & K^{1,0} & \dots & 0 \\ \vdots & \vdots & \ddots & 0 \\ K^{N,N} & K^{N,N-1} & \dots & K^{N,0} \end{bmatrix},
\end{equation*}
resulting in the closed-loop trajectory $\mathbf{x} = \left(\mathcal{ZA} + \mathcal{ZB}\mathbf{K}\right)\mathbf{x} + \bm{\delta} = \left(I - \mathcal{ZA} - \mathcal{ZB}\mathbf{K}\right)^{-1}\bm{\delta}$. The \emph{system response} is then defined by the closed-loop map~$\bm{\Phi}: \bm{\delta} \to \left(\mathbf{x}, \mathbf{u}\right)$ as
\begin{equation}\label{eq:SLP}
\begin{bmatrix} \mathbf{x} \\ \mathbf{u} \end{bmatrix} = \begin{bmatrix} \left(I - \mathcal{ZA} - \mathcal{ZB}\mathbf{K}\right)^{-1} \\ \mathbf{K}\left(I - \mathcal{ZA} - \mathcal{ZB}\mathbf{K}\right)^{-1} \end{bmatrix} \bm{\delta} = \begin{bmatrix} \bm{\Phi}_\mathbf{x} \\ \bm{\Phi}_\mathbf{u} \end{bmatrix} \bm{\delta} = \bm{\Phi} \bm{\delta}.
\end{equation}
The maps $\bm{\Phi}_\mathbf{x}$,\, $\bm{\Phi}_\mathbf{u}$ have a block-lower-triangular structure and completely define the behavior of the closed-loop system with feedback controller $\mathbf{K}$. Using these maps as the parametrization of the system behavior allows us to recast the optimal control problem in terms of $\bm{\Phi}_\mathbf{x}$, $\bm{\Phi}_\mathbf{u}$ instead of the state feedback gain $\mathbf{K}$ by exploiting the following theorem.

\begin{theorem}{(Theorem 2.1 in \cite{Anderson2019})}\label{theorem:SLS}
Over the horizon $N$, the system dynamics~\eqref{eq:dynamics} with block-lower-triangular state feedback law $\mathbf{K}$ defining the control action as $\mathbf{u} = \mathbf{Kx}$, the following are true:
\begin{enumerate}
	\item the affine subspace defined by
	\begin{equation}\label{SLS:affine-halfspace}\begin{bmatrix} I - \mathcal{ZA} & -\mathcal{ZB}\end{bmatrix} \begin{bmatrix} \bm{\Phi}_x \\ \bm{\Phi}_u \end{bmatrix} = I \end{equation} parametrizes all possible system responses~\eqref{eq:SLP},
	\item for any block-lower-triangular matrices $\bm{\Phi}_x$, $\bm{\Phi}_u$ satisfying~\eqref{SLS:affine-halfspace}, the controller $\mathbf{K} = \bm{\Phi}_u \bm{\Phi}_x^{-1}$ achieves the desired system response.
\end{enumerate}
\end{theorem}

\subsection{Robust Model Predictive Control (MPC)}
We define the optimal controller for system~\eqref{eq:dynamics} subject to~\eqref{eq:constraints} via a robust MPC formulation~\cite{Rawlings2009}, where we consider the quadratic cost $\sum_{k=0}^{N-1} \norm{x_k}_{Q_k}^2 + \norm{u_k}_{R_k}^2,$ with $Q_k, R_k$ the state and input weights at time step $k$, respectively, and $N$ the prediction horizon. For simplicity, we focus on quadratic costs, however the results can be extended to other cost functions fulfilling the standard MPC stability assumptions~\cite{Rawlings2009}. Using the compact system definition in~\eqref{eq:stacked_dynamics}, the robust MPC problem is then defined as
\begin{subequations}\label{MPC:generic}
	\begin{align}
		\min_{\bm{\pi}} \;\; & \mathbf{x}^\top \mathcal{Q} \mathbf{x} + \mathbf{u}^\top \mathcal{R} \mathbf{u}, \label{MPC:cost}\\
		\textrm{s.t. }        & \mathbf{x} = \mathcal{ZA}\mathbf{x} + \mathcal{ZB}\mathbf{u} + \bm{\delta} \\
		& \mathbf{x}^{:N} \in \mathcal{X}^N\!\!, \ \mathbf{x}^{N} \in \mathcal{X}_f, \ \mathbf{u} \in \mathcal{U}^N\!\!, \quad \forall\mathbf{w} \in \mathcal{W}^{N-1} \label{MPC:constraints}\\
        & \mathbf{u} = \bm{\pi}(\mathbf{x}, \mathbf{u}), \ \mathbf{x}^0 = x_k
	\end{align}
\end{subequations}
where \\
${\mathcal{Q} = \diag(Q_0, \dots, Q_{N-1}, P)}$, ${\mathcal{R} = \diag(R_0, \dots, R_{N-1}, 0)}$, with $P$ a suitable terminal cost matrix, $\mathcal{X}_f$ a suitable terminal set, and $\bm{\pi} = [ \pi^0, \dots, \pi^N ]^\top$ a vector of parametrized feedback policies. Available robust MPC methods can be classified according to their policy parameterization and constraint handling approach. Two of the most common policies are the tube policy~\cite{Mayne2005} and the disturbance feedback policy~\cite{Lofberg2003}, defined as
\begin{align}
\bm{\pi}^{tube} (\mathbf{x}, \mathbf{u}) &=  \mathbf{K}\left(\mathbf{x}-\mathbf{z}\right) + \mathbf{v}, \label{tube-policy}\\
\bm{\pi}^{df} (\mathbf{x}, \mathbf{u}) &= \mathbf{Mw} + \mathbf{v}, \label{df-policy}
\end{align}
where $\mathbf{z} = [ z_0, \dots, z_{N} ]^\top$, $\mathbf{v} = [ v_0, \dots, v_{N} ]^\top$, and
\begin{equation}\label{eq:M-structure}
\mathbf{K} \!=\! \!\begin{bmatrix} K & & \\ & \ddots & \\ & & K \end{bmatrix}\!, \;\
\mathbf{M} \!=\! \!\begin{bmatrix} 0 & \hspace{-.3em}\dots & \hspace{-.3em}0 \\ M_{1,0} & \hspace{-.3em}\dots & \hspace{-.3em}0 \\ \vdots & \hspace{-.3em}\ddots  & \hspace{-.3em}\vdots \\ M_{N,0} & \hspace{-.3em}\dots  & \hspace{-.3em}M_{N,N-1} \end{bmatrix}\!.
\end{equation}
The tube policy parameterization~\eqref{tube-policy} is based on splitting the system dynamics~\eqref{eq:stacked_dynamics} into nominal dynamics and error dynamics, i.e.
\begin{align}
\mathbf{z} &= \mathcal{ZA}\mathbf{z} + \mathcal{ZB}\mathbf{v}, \label{eq:nominal-dynamics}
 \\
\mathbf{x} - \mathbf{z} &= \mathbf{e} = \left(\mathcal{ZA} + \mathcal{ZB}\mathbf{K}\right)\mathbf{e} + \bm{\delta}. \label{eq:error-dynamics}
\end{align}
This allows for recasting~\eqref{MPC:generic} in the nominal variables $\mathbf{z},\, \mathbf{v}$ instead of the system variables $\mathbf{x},\, \mathbf{u}$, while imposing tightened constraints on the nominal variables. One way to perform the constraint tightening, which we will focus on in this paper, is via reachable sets for the error dynamics. In this paper, we will refer to this approach as \tube-MPC, which was first introduced in~\cite{Chisci2001}. For an overview of other tube-based methods or robust MPC in general, see e.g.~\cite{Rawlings2009}.

\section{A SYSTEM LEVEL APPROACH TO TUBE-BASED MODEL PREDICTIVE CONTROL}\label{sec:SL-perspective}
In this section, we present a new perspective on disturbance feedback MPC~(\df-MPC) by utilizing the SLP. In particular, we show the equivalence of SLP and \df-MPC, which includes the tube policy~\eqref{tube-policy} as a subclass~\cite{Rakovic2012}, and discuss the implications that arise.

Consider the disturbance feedback policy~\eqref{df-policy}. We define the convex set of admissible $(\mathbf{M}, \mathbf{v})$ as
\begin{equation*}
\Pi^{df}_N(x_0) = \left\lbrace (\mathbf{M}, \mathbf{v}) \middle| \!\begin{array}{l} \mathbf{M} \text{ structured as in~\eqref{eq:M-structure}} \\ \mathbf{x}^{:N} \in \mathcal{X}^N, \ \mathbf{u}^{df} \in \mathcal{U}^N, \\ \mathbf{x}^{N} \in \mathcal{X}_f, \ \forall \mathbf{w} \in \mathcal{W}^{N-1} \end{array} \right\rbrace,
\end{equation*}
and the set of initial states $x_0$ for which an admissible control policy of the form~\eqref{df-policy} exists is given by $X^{df}_N = \{x_0 \mid \Pi^{df}_N(x_0) \neq \emptyset \}$. Similarly, we define the convex set of admissible~$(\bm{\Phi}_\mathbf{x}, \bm{\Phi}_\mathbf{u})$ as
\begin{equation*}
\Pi^{SLP}_N(x_0) = \left\lbrace \! (\bm{\Phi}_\mathbf{x}, \bm{\Phi}_\mathbf{u}) \middle| \!\begin{array}{l} \bm{\Phi}_\mathbf{x}, \bm{\Phi}_\mathbf{u} \text{ satisfy~\eqref{SLS:affine-halfspace} and} \\ \text{are block-lower-triangular,} \\ \bm{\Phi}_\mathbf{x}^{:N,:}\bm{\delta} \in \mathcal{X}^N, \ \bm{\Phi}_\mathbf{u}\bm{\delta} \in \mathcal{U}^N, \\ \bm{\Phi}_\mathbf{x}^{N,:}\bm{\delta} \in \mathcal{X}_f, \ \forall \mathbf{w} \in \mathcal{W}^{N-1} \end{array} \right\rbrace,
\end{equation*}
and the set of initial states $x_0$ for which an admissible control policy $\mathbf{u}^{SLP} = \bm{\Phi}_\mathbf{u}\bm{\Phi}_\mathbf{x}^{-1}\mathbf{x}$ exists, as $X^{SLP}_N = \{x_0 \mid \Pi^{SLP}_N(x_0) \neq \emptyset \}$. To show the equivalence between MPC and SLP trajectories, we will rely on Lemma~\ref{lemma:MPC_to_SLS} and we formalize the equivalence in Theorem~\ref{theorem:equivalence}.

\begin{lemma}\label{lemma:MPC_to_SLS}
Consider the dynamics~\eqref{eq:stacked_dynamics} as a function of the initial state $x_0$
\begin{equation}\label{eq:function_of_x0}
\mathbf{x} = \mathbf{A}x_0 +\mathbf{B}\mathbf{u} + \mathbf{Ew},
\end{equation}
with $\mathbf{A}$, $\mathbf{B}$, and $\mathbf{E}$ defined as in Appendix~\ref{apx:proof-lemma}. Then, the following relation holds for any $\mathbf{u}$, any $\mathbf{w}$, and any admissible $(\bm{\Phi}_\mathbf{x}, \bm{\Phi}_\mathbf{u})$:
\begin{align}
\bm{\Phi}_\mathbf{x} &= \begin{bmatrix} \mathbf{A} & \mathbf{E} \end{bmatrix} + \mathbf{B}\bm{\Phi}_\mathbf{u}. \label{eq:rewritten_affine_halfspace}
\end{align}
\end{lemma}
\begin{proof}
See Appendix~\ref{apx:proof-lemma}.
\end{proof}

\begin{theorem}\label{theorem:equivalence}
Given any initial state $x_0 \in X^{df}_N$,~$(\mathbf{M}, \mathbf{v}) \in \Pi^{df}_N(x_0)$, and some disturbance sequence~$\mathbf{w} \in \mathcal{W}^{N-1}$, there exist $(\bm{\Phi}_\mathbf{x}, \bm{\Phi}_\mathbf{u}) \in \Pi^{SLP}_N(x_0)$ such that $\mathbf{x}^{df} = \mathbf{x}^{SLP}$ and $\mathbf{u}^{df} = \mathbf{u}^{SLP}$. The same statement holds in the opposite direction for any $x_0 \in X^{SLP}_N$,~$(\bm{\Phi}_\mathbf{x}, \bm{\Phi}_\mathbf{u}) \in \Pi^{SLP}_N(x_0)$, and disturbance~$\mathbf{w} \in \mathcal{W}^{N-1}$, therefore $X^{df}_N = X^{SLP}_N$.
\end{theorem}
\begin{proof}
We prove $X^{df}_N = X^{SLP}_N$ by proving the two set inclusions $X^{df}_N \subseteq X^{SLP}_N$ and $X^{SLP}_N \subseteq X^{df}_N$ separately. The detailed steps of the two parts are stated in Table~\ref{table:proof}.
\paragraph{$X^{df}_N \subseteq X^{SLP}_N$} Given $x_0 \in X^{df}_N$, an admissible $(\mathbf{M}, \mathbf{v}) \in \Pi^{df}_N(x_0)$ exists by definition. Using~\eqref{eq:function_of_x0}, we write the input and state trajectories for a given $\mathbf{w} \in \mathcal{W}^{N-1}$ as $\mathbf{u}^{df} = \mathbf{Mw} + \mathbf{v}$ and $\mathbf{x}^{df} = \mathbf{A}x_0 + \mathbf{B}\mathbf{u}^{df} + \mathbf{Ew}$, respectively. In order to show the existence of a corresponding $(\bm{\Phi}_\mathbf{x}, \bm{\Phi}_\mathbf{u})$, we choose $(\bm{\Phi}_\mathbf{x}, \bm{\Phi}_\mathbf{u})$ according to~\eqref{proof:SLP_param_choice} in Table~\ref{table:proof} and~$\Phi_v$ such that $\Phi_v x_0 = \mathbf{v}$ holds. Then, the derivations in Table~\ref{table:proof} show that the input and state trajectories of \df-MPC and SLP are equivalent. Therefore, ${x_0 \in X^{df}_N \implies x_0 \in X^{SLP}_N}$.
\paragraph{$X^{SLP}_N \subseteq X^{df}_N$} Given $x_0 \in X^{SLP}_N$, an admissible $(\bm{\Phi}_\mathbf{x}, \bm{\Phi}_\mathbf{u})$ exists by definition. Using~\eqref{eq:SLP} and a given $\mathbf{w} \in \mathcal{W}^{N-1}$ we write the state and input trajectories as $\mathbf{x}^{SLP} = \bm{\Phi}_\mathbf{x}\bm{\delta}$ and $\mathbf{u}^{SLP} = \bm{\Phi}_\mathbf{u}\bm{\delta}$, respectively. If $(\mathbf{M}, \mathbf{v})$ is chosen according to~\eqref{proof:df_param_choice} in Table~\ref{table:proof}, then these choices are admissible and the input and state trajectories for both \df-MPC and SLP are equivalent. Therefore, ${x_0 \in X^{SLP}_N \implies x_0 \in X^{df}_N}$.
\end{proof}

\begin{table}[h!]
\centering
\caption{Derivation of the equivalence between SLP and \df-MPC.}\label{table:proof}
\begin{tabular}{@{}cc@{}}
\toprule
$X^{df}_N \subseteq X^{SLP}_N$ & $X^{SLP}_N \subseteq X^{df}_N$ \\ \midrule[.1em]
{$\!\begin{aligned} \bm{\Phi}_\mathbf{x} &= \begin{bmatrix} \mathbf{A} & \hspace{-.5em}\mathbf{E} \end{bmatrix} + \mathbf{B}\bm{\Phi}_\mathbf{u} \\ \bm{\Phi}_\mathbf{u} &= \begin{bmatrix} \Phi_v & \hspace{-.5em}\mathbf{M} \end{bmatrix} \end{aligned} \;\;\eqnum\label{proof:SLP_param_choice}$} & {$\!\begin{aligned} \mathbf{v} = \bm{\Phi}_\mathbf{u}^{:,0} x_0, \; \mathbf{M} = \bm{\Phi}_\mathbf{u}^{:,1:} \end{aligned} \;\;\eqnum\label{proof:df_param_choice}$} \\ \cmidrule(lr){1-1} \cmidrule(lr){2-2} 
{$\!\begin{aligned} \mathbf{u}^{SLP} \!\!&= \bm{\Phi}_\mathbf{u} \bm{\delta} \mathrel{\overset{\makebox[0pt]{\mbox{\normalfont\tiny\sffamily \eqref{proof:SLP_param_choice}}}}{=}} \begin{bmatrix} \Phi_v & \hspace{-.5em}\mathbf{M} \end{bmatrix} \bm{\delta} \\ &= \Phi_v x_0 + \mathbf{Mw} \\ &= \mathbf{v} + \mathbf{Mw} = \mathbf{u}^{df} \\ \mathbf{x}^{SLP} \!\!&= \bm{\Phi}_\mathbf{x} \bm{\delta} \\ &\mathrel{\overset{\makebox[0pt]{\mbox{\normalfont\tiny\sffamily \eqref{proof:SLP_param_choice}}}}{=}} \begin{bmatrix} \mathbf{A} & \hspace{-.5em}\mathbf{E} \end{bmatrix} \bm{\delta} + \mathbf{B}\bm{\Phi}_\mathbf{u} \bm{\delta} \\ &= \mathbf{A}x_0 \!+\! \mathbf{Ew} \!+\! \mathbf{B}\mathbf{u}^{df} = \mathbf{x}^{df} \end{aligned}$} & {$\!\begin{aligned} \mathbf{u}^{df} \!\!&= \mathbf{Mw} + \mathbf{v} \\ &\mathrel{\overset{\makebox[0pt]{\mbox{\normalfont\tiny\sffamily \eqref{proof:df_param_choice}}}}{=}} \bm{\Phi}_\mathbf{u}^{:,1:}\mathbf{w} + \bm{\Phi}_\mathbf{u}^{:,0} x_0 \\ &= \bm{\Phi}_\mathbf{u}\bm{\delta} = \mathbf{u}^{SLP} \\ \mathbf{x}^{df} \!\!&= \mathbf{A}x_0 \!+\! \mathbf{Ew} \!+\! \mathbf{B}(\mathbf{Mw} \!+\! \mathbf{v}) \\ &= \begin{bmatrix} \mathbf{A} & \hspace{-.5em}\mathbf{E} \end{bmatrix} \bm{\delta} + \mathbf{B}\bm{\Phi}_\mathbf{u} \bm{\delta} \\ &\mathrel{\overset{\makebox[0pt]{\mbox{\normalfont\tiny\sffamily \eqref{eq:rewritten_affine_halfspace}}}}{=}} \bm{\Phi}_\mathbf{x} \bm{\delta} = \mathbf{x}^{SLP} \end{aligned}$} \\ \bottomrule 
\end{tabular}
\end{table}

\begin{remark}
Theorem~\ref{theorem:equivalence} is also valid for linear time-varying~(LTV) systems without any modifications to the proof, by adapting the definitions of $\mathcal{ZA}$ and $\mathcal{ZB}$ accordingly.
\end{remark}

Leveraging Theorem~\ref{theorem:equivalence}, the following insights can be derived: Compared with \df-MPC, the SLP formulation offers a direct handle on the state trajectory, which allows for directly imposing constraints on the state and simplifies the numerical implementation. Additionally, computational efficiency is increased by recently developed solvers for SLP problems based on dynamic programming~\cite{Tseng2020a} and ready-to-use computational frameworks like SLSpy~\cite{Tseng2020b}. The SLP formulation can thereby address one of the biggest limitations of \df-MPC, i.e. the computational effort required. SLP problems also facilitate the computation of a distributed controller due to the parameterization by system responses $\bm{\Phi}_\mathbf{x}$,\, $\bm{\Phi}_\mathbf{u}$, on which a distributed structure can be imposed through their support~\cite{Wang2018a}. This offers a direct method to also distribute \df-MPC problems. We will leverage Theorem~\ref{theorem:equivalence} to derive an improved \tube-MPC formulation as the main result of this paper in the following section.

\section{SYSTEM LEVEL TUBE-MPC (SLTMPC)}\label{sec:SLTMPC}
We first analyze the effect of imposing additional structure on the SLP, before deriving the SLTMPC formulation. Finally, we show that the SLTMPC formulation emerges naturally from an extended SLP formulation.

\subsection{Diagonally Restricted System Responses}
Following the analysis in~\cite{Rakovic2012}, it can be shown that disturbance feedback MPC is equivalent to a time-varying version of \tube-MPC. The same insight can be derived from the SLP by defining the nominal state and input as $\mathbf{z} \coloneqq \bm{\Phi}_\mathbf{x}^{:,0} x_0$ and $\mathbf{v} \coloneqq \bm{\Phi}_\mathbf{u}^{:,0} x_0$, respectively, and solving $\mathbf{x} = \bm{\Phi}_\mathbf{x}\bm{\delta}$ for the disturbance trajectory $\mathbf{w}$:
\begin{align*}
\mathbf{x} - \mathbf{z} &= \bm{\Phi}_\mathbf{x}^{:,1:}\mathbf{w}, \\
\begin{bmatrix} 0 \\ \mathbf{x}^{1:} - \mathbf{z}^{1:} \end{bmatrix} &= \begin{bmatrix} \bm{0} \\ \bm{\Phi}_\mathbf{x}^{1:,1:}\mathbf{w} \end{bmatrix}, \\
&\Rightarrow \mathbf{w} = \left[\bm{\Phi}_\mathbf{x}^{1:,1:}\right]^{-1} \left( \mathbf{x}^{1:} - \mathbf{z}^{1:} \right).
\end{align*}
Plugging this definition of $\mathbf{w}$ into $\mathbf{u} = \bm{\Phi}_\mathbf{u}\bm{\delta}$, yields $$ \mathbf{u} = \mathbf{v} + \bm{\Phi}_\mathbf{u}^{:,1:}\left[\bm{\Phi}_\mathbf{x}^{1:,1:}\right]^{-1} \left( \mathbf{x}^{1:} - \mathbf{z}^{1:} \right) = \mathbf{v} + \mathbf{K}^{:,1:} \left( \mathbf{x}^{1:} - \mathbf{z}^{1:} \right), $$ which reveals the inherent tube structure of the SLP and consequently \df-MPC, with the time-varying tube controller~$\mathbf{K}$. The main drawback of this formulation is its computational complexity, since the number of optimization variables grows quadratically in the horizon length~$N$, which motivates the derivation of a computationally more efficient \df-MPC formulation. Inspired by the infinite-horizon SLP formulation~\cite{Anderson2019}, we restrict~$\bm{\Phi}_\mathbf{x}$ and~$\bm{\Phi}_\mathbf{u}$ to have constant diagonal entries\footnote{In the infinite-horizon formulation, the diagonals are restricted in order to comply with the frequency-domain interpretation of the closed-loop system. Commonly, a finite impulse response (FIR) constraint is additionally imposed in this setting, requiring the closed-loop dynamics to approach an equilibrium after the FIR horizon, which we do not need here. For more details see~\cite{Anderson2019}.}, which corresponds to restricting~$\mathbf{K}$ to also have constant diagonal entries and therefore results in a more restrictive controller parameterization. This results in a tube formulation where the number of optimization variables is linear in the horizon length, alleviating some of the computational burden:
\begin{equation}\label{eq:mb_structure}
\bar{\bm{\Phi}}_\mathbf{x} \!=\!\!\begin{bmatrix} \hspace{-.1em}\Phi_x^0 & \hspace{-.3em}\dots & \hspace{-.3em}\dots & \hspace{-.3em}0 \\ \hspace{-.1em}\Phi_x^1  & \hspace{-.3em} \Phi_x^0 & \hspace{-.3em}\dots & \hspace{-.3em}0 \\ \hspace{-.1em}\vdots & \hspace{-.3em}\ddots & \hspace{-.3em}\ddots & \hspace{-.3em}\vdots \\ \hspace{-.1em}\Phi_x^N & \hspace{-.3em}\dots & \hspace{-.3em}\Phi_x^1 & \hspace{-.3em}\Phi_x^0 \end{bmatrix}\!\!, \:
\bar{\bm{\Phi}}_\mathbf{u} \!=\!\!\begin{bmatrix} \hspace{-.1em}\Phi_u^0 & \hspace{-.3em}\dots & \hspace{-.3em}\dots & \hspace{-.3em}0 \\ \hspace{-.1em}\Phi_u^1  & \hspace{-.3em} \Phi_u^0 & \hspace{-.3em}\dots & \hspace{-.3em}0 \\ \hspace{-.1em}\vdots & \hspace{-.3em}\ddots & \hspace{-.3em}\ddots & \hspace{-.3em}\vdots \\ \hspace{-.1em}\Phi_u^N & \hspace{-.3em}\dots & \hspace{-.3em}\Phi_u^1 & \hspace{-.3em}\Phi_u^0 \end{bmatrix}\!\!.
\end{equation}
Enforcing the structure~\eqref{eq:mb_structure} also alters the equivalence relation established by Theorem~\ref{theorem:equivalence}. If we interpret the imposed structure in a \tube-MPC framework, this corresponds to both the nominal and error system being driven by the controller~$\bar{\mathbf{K}}$, since~$\bar{\bm{\Phi}}_\mathbf{x}, \bar{\bm{\Phi}}_\mathbf{u}$ not only define the tube controller but also the behavior of the nominal system, i.e. $\bar{\mathbf{z}}=\bar{\bm{\Phi}}_\mathbf{x}^{:,0}x_0$ and $\bar{\mathbf{v}}=\bar{\bm{\Phi}}_\mathbf{u}^{:,0}x_0$. To illustrate this point, we exemplify it on the LTI system~\eqref{eq:dynamics} subject to a constant tube controller $\bar{K}$. This results in the input ${u_i = \bar{K} ( x_i - z_i) + v_i}$ and thus in a constant block-diagonal matrix ${\bar{\mathbf{K}} = \diag (\bar{K}, \dots, \bar{K})}$ over the horizon $N$. The associated system responses are computed using~\eqref{eq:SLP}, resulting in
\begin{equation*}
\bar{\bm{\Phi}}_\mathbf{x}^{\bar{K}} \!=\!\!\begin{bmatrix} \hspace{-.1em}I & \hspace{-.3em}\dots & \hspace{-.3em}\dots & \hspace{-.3em}0 \\ \hspace{-.1em}A_{\bar{K}}  & \hspace{-.3em} I & \hspace{-.3em}\dots & \hspace{-.3em}0 \\ \hspace{-.1em}\vdots & \hspace{-.3em}\ddots & \hspace{-.3em}\ddots & \hspace{-.3em}\vdots \\ \hspace{-.1em}A_{\bar{K}}^N & \hspace{-.3em}\dots & \hspace{-.3em}A_{\bar{K}} & \hspace{-.3em}I \end{bmatrix}\!, \;
\bar{\bm{\Phi}}_\mathbf{u}^{\bar{K}} \!=\!\!\begin{bmatrix} \hspace{-.1em}\bar{K} & \hspace{-.3em}\dots & \hspace{-.3em}\dots & \hspace{-.3em}0 \\ \hspace{-.1em}\bar{K}A_{\bar{K}}  & \hspace{-.3em} \bar{K} & \hspace{-.3em}\dots & \hspace{-.3em}0 \\ \hspace{-.1em}\vdots & \hspace{-.3em}\ddots & \hspace{-.3em}\ddots & \hspace{-.3em}\vdots \\ \hspace{-.1em}\bar{K}A_{\bar{K}}^N & \hspace{-.3em}\dots & \hspace{-.3em}\bar{K}A_{\bar{K}} & \hspace{-.3em}\bar{K} \end{bmatrix}\!,
\end{equation*}
where $A_{\bar{K}} = A+B\bar{K}$ and $A, \, B$ as in~\eqref{eq:dynamics}. Given these system responses, the states and inputs over the horizon $N$ are computed as $\mathbf{x} = \bar{\bm{\Phi}}_\mathbf{x}^{\bar{K}}\bm{\delta}$ and $\mathbf{u} = \bar{\bm{\Phi}}_\mathbf{u}^{\bar{K}}\bm{\delta}$, respectively, or in convolutional form for each component of the sequence
\begin{align*}
x_i \!&=\! \Phi_x^{\bar{K},i} x_0 \!+\! \sum_{j=0}^{i-1} \Phi_x^{\bar{K},i-1-j} w_j = A_{\bar{K}}^i x_0 \!+\! \sum_{j=0}^{i-1} A_{\bar{K}}^{i-1-j} w_j,\\
u_i \!&=\! \Phi_u^{\bar{K},i} x_0 \!+\! \sum_{j=0}^{i-1} \Phi_u^{\bar{K},i-1-j} w_j \!=\! \bar{K}\!A_{\bar{K}}^i x_0 \!+\! \bar{K}\!\sum_{j=0}^{i-1} \!A_{\bar{K}}^{i-1-j} w_j.
\end{align*}
The nominal states and nominal inputs are then defined by the relations
\begin{equation}\label{eq:nominalSys-example}
z_i = \Phi_x^{\bar{K},i} x_0 = A_{\bar{K}}^i x_0, \quad v_i = \Phi_u^{\bar{K},i} x_0 = \bar{K}A_{\bar{K}}^i x_0.
\end{equation}
Therefore,
\begin{equation}\label{eq:tube-example}
x_i = z_i + \sum_{j=0}^{i-1} A_{\bar{K}}^{i-1-j} w_j, \quad u_i = v_i + \bar{K} \left( x_i - z_i \right),
\end{equation}
where we use the fact that $\sum_{j=0}^{i-1} A_{\bar{K}}^{i-1-j} w_j = x_i - z_i$ to rewrite the input relation. As in tube-MPC~\cite[Section~3.2]{Chisci2001}, the uncertain state is described via the nominal state and the evolution of the error system $\sum_{j=0}^{i-1} A_{\bar{K}}^{i-1-j} w_j$ under the tube controller $\bar{K}$. In contrast to the \tube-MPC formulation, however, $z_i,\, v_i$ are not free variables here, but are governed by the equations~\eqref{eq:nominalSys-example}, i.e. $z_{i+1} = (A+B\bar{K})z_{i}$, which is an LTI system with feedback controller $\bar{K}$. Therefore, both the error system and the nominal system are driven by the same controller, which is generally suboptimal.

The following section derives an extended SLP that allows for optimizing the nominal system dynamics in the corresponding \tube-MPC interpretation independent of the error dynamics. This is achieved by adding the auxiliary decision variables~$\phi_z^i,\, \phi_v^i$ to the nominal variables $z_i$ and $v_i$, i.e. $\bar{z}_i = \Phi_x^{i} x_0 + \phi_z^i$, $\bar{v}_i = \Phi_u^{i} x_0 + \phi_v^i$. In vector notation, this results in
\begin{equation}\label{eq:extended_SLP}
\bar{\mathbf{z}} = \bar{\bm{\Phi}}_\mathbf{x}^{:,0}x_0 + \bm{\phi}_\mathbf{z}, \qquad \bar{\mathbf{v}} = \bar{\bm{\Phi}}_\mathbf{u}^{:,0}x_0 + \bm{\phi}_\mathbf{v}.
\end{equation}

\subsection{Affine System Level Parameterization}
We extend the SLP in order to capture formulations of the form~\eqref{eq:extended_SLP}. The following derivations are stated for general system responses $\bm{\Phi}_\mathbf{x}$, $\bm{\Phi}_\mathbf{u}$ but naturally extend to the restricted system responses $\bar{\bm{\Phi}}_\mathbf{x}$, $\bar{\bm{\Phi}}_\mathbf{u}$ without any modifications. We expand the state~$\tilde{\mathbf{x}} = \left[ 1\; \mathbf{x}\right]^\top$, the disturbance~${\tilde{\bm{\delta}} = \left[ 1\; \bm{\delta}\right]^\top}$, and appropriately modify the dynamics~\eqref{eq:stacked_dynamics} to obtain
\begin{equation}\label{eq:affine_dynamics}
\tilde{\mathbf{x}} = \mathcal{Z}\begin{bmatrix} 0 & \bm{0} \\ \bm{0} & \mathcal{A} \end{bmatrix}\tilde{\mathbf{x}} + \mathcal{Z}\begin{bmatrix} \bm{0} \\ \mathcal{B} \end{bmatrix}\mathbf{u} + \tilde{\bm{\delta}}
= \mathcal{Z}\tilde{\mathcal{A}}\tilde{\mathbf{x}} + \mathcal{Z}\tilde{\mathcal{B}}\mathbf{u} + \tilde{\bm{\delta}}.
\end{equation}
Then, the corresponding \emph{affine system responses} $\tilde{\bm{\Phi}}_\mathbf{x}, \tilde{\bm{\Phi}}_\mathbf{u}$, assuming an affine state feedback control law $\mathbf{u} = \tilde{\mathbf{K}}\tilde{\mathbf{x}} = \mathbf{Kx} + \mathbf{k}$, are defined as
\begin{equation}\label{eq:affine_structure}
\tilde{\bm{\Phi}}_\mathbf{x} = \begin{bmatrix} 1 & \bm{0} \\ \bm{\phi}_\mathbf{z} & \bm{\Phi}_\mathbf{e} \end{bmatrix}, \qquad \tilde{\bm{\Phi}}_\mathbf{u} = \begin{bmatrix} \bm{\phi}_\mathbf{v} & \bm{\Phi}_\mathbf{k} \end{bmatrix},
\end{equation}
where $\bm{\phi}_\star$ and $\bm{\Phi}_\star$ denote vector-valued and matrix-valued quantities, respectively. Given the particular structure of the affine system responses, we can extend Theorem~\ref{theorem:SLS} to the more general affine case.

\begin{theorem}\label{theorem:affine_SLS}
Consider the system dynamics~\eqref{eq:affine_dynamics} over a horizon $N$ with block-lower-triangular state feedback law~$\tilde{\mathbf{K}}$ defining the control action as $\mathbf{u} = \tilde{\mathbf{K}}\tilde{\mathbf{x}}$. The following statements are true:
\begin{enumerate}
	\item the affine subspace defined by
	\begin{equation}\label{affine_SLS:affine-halfspace}\begin{bmatrix} I - \mathcal{Z}\tilde{\mathcal{A}} & -\mathcal{Z}\tilde{\mathcal{B}}\end{bmatrix} \begin{bmatrix} \tilde{\bm{\Phi}}_\mathbf{x} \\ \tilde{\bm{\Phi}}_\mathbf{u} \end{bmatrix} = I \end{equation} parametrizes all possible affine system responses $\tilde{\bm{\Phi}}_\mathbf{x}$,~$\tilde{\bm{\Phi}}_\mathbf{u}$ with structure~\eqref{eq:affine_structure},
	\item for any block-lower-triangular matrices $\tilde{\bm{\Phi}}_\mathbf{x}$, $\tilde{\bm{\Phi}}_\mathbf{u}$ satisfying~\eqref{affine_SLS:affine-halfspace}, the controller $\tilde{\mathbf{K}} = \tilde{\bm{\Phi}}_\mathbf{u} \tilde{\bm{\Phi}}_\mathbf{x}^{-1}$ achieves the desired affine system response.
\end{enumerate}
\end{theorem}
\begin{proof}
\paragraph{Proof of 1)} Follows trivially from the proof of Theorem 2.1 in~\cite{Anderson2019}.
\paragraph{Proof of 2)} As in the proof of Theorem 2.1 in~\cite{Anderson2019}, $\tilde{\bm{\Phi}}_\mathbf{x}^{-1}$ exists since $\tilde{\bm{\Phi}}_\mathbf{x}$ admits a finite Neumann series\footnote{The inverse of a matrix $(I-A)^{-1}$ can be expressed as an infinite Neumann series, iff the matrix $(I-A)$ is stable. However, if $A$ is nilpotent the stability requirement is omitted and the inverse is defined by the truncated Neumann series~\cite{Stewart1998}.}. In order to show $\tilde{\mathbf{x}} = \tilde{\bm{\Phi}}_\mathbf{x}\tilde{\bm{\delta}}$, the following relation needs to hold
\begin{equation}\label{proof:condition}
\tilde{\bm{\Phi}}_\mathbf{x} = \left( I - \mathcal{Z}\tilde{\mathcal{A}} - \mathcal{Z}\tilde{\mathcal{B}}\tilde{\bm{\Phi}}_\mathbf{u}\tilde{\bm{\Phi}}_\mathbf{x}^{-1}\right)^{-1}\!\!\!\!.
\end{equation}
Plugging in the affine system responses and after some algebraic manipulations this relation becomes
\begin{equation*}
 \tilde{\bm{\Phi}}_\mathbf{x} =\! \begin{bmatrix} 1 & \bm{0} \\ -\mathcal{ZB}\left(\bm{\phi}_\mathbf{v} - \bm{\Phi}_\mathbf{k}\bm{\Phi}_\mathbf{e}^{-1}\bm{\phi}_\mathbf{z}\right) & I - \mathcal{ZA} - \mathcal{ZB}\bm{\Phi}_\mathbf{k}\bm{\Phi}_\mathbf{e}^{-1} \end{bmatrix}^{-1}\!\!\!\!\!\!\!,
\end{equation*}
which holds if the following two equations hold
\begin{align}
\bm{\Phi}_\mathbf{e} &= \left(I - \!\mathcal{ZA} - \mathcal{ZB}\bm{\Phi}_\mathbf{k}\bm{\Phi}_\mathbf{e}^{-1}\right)^{-1}\!\!, \label{eq:feedback_part} \\
\bm{\phi}_\mathbf{z} &= \bm{\Phi}_\mathbf{e}\mathcal{ZB}\left(\bm{\phi}_\mathbf{v} - \bm{\Phi}_\mathbf{k}\bm{\Phi}_\mathbf{e}^{-1}\bm{\phi}_\mathbf{z}\right)\!. \label{eq:affine_part}
\end{align}
Equation~\eqref{eq:feedback_part} is fulfilled if $\bm{\Phi}_\mathbf{e}$ and $\bm{\Phi}_\mathbf{k}$ satisfy
\begin{equation}\label{affine_SLS:affine_halfspace_feedback_part}
\begin{bmatrix} I - \mathcal{ZA} & -\mathcal{ZB}\end{bmatrix} \begin{bmatrix} \bm{\Phi}_\mathbf{e} \\ \bm{\Phi}_\mathbf{k} \end{bmatrix} = I
\end{equation}
as in~\cite{Anderson2019}, then~\eqref{eq:affine_part} can be rewritten as
\begin{align}
\bm{\Phi}_\mathbf{e}^{-1}\bm{\phi}_\mathbf{z} &= \mathcal{ZB}\left(\bm{\phi}_\mathbf{v} - \bm{\Phi}_\mathbf{k}\bm{\Phi}_\mathbf{e}^{-1}\bm{\phi}_\mathbf{z}\right), \nonumber \\
\left(I + \mathcal{ZB}\bm{\Phi}_\mathbf{k}\right)\bm{\Phi}_\mathbf{e}^{-1}\bm{\phi}_\mathbf{z} &= \mathcal{ZB}\bm{\phi}_\mathbf{v}, \nonumber \\
\left(I - \mathcal{ZA}\right)\bm{\Phi}_\mathbf{e}\bm{\Phi}_\mathbf{e}^{-1}\bm{\phi}_\mathbf{z} &= \mathcal{ZB}\bm{\phi}_\mathbf{v}, \nonumber \\
\begin{bmatrix} I - \mathcal{ZA} & -\mathcal{ZB}\end{bmatrix} \begin{bmatrix} \bm{\phi}_\mathbf{z} \\ \bm{\phi}_\mathbf{v} \end{bmatrix} &= \bm{0}, \label{affine_SLS:affine_halfspace_affine_part}
\end{align}
where we use~\eqref{affine_SLS:affine_halfspace_feedback_part} to eliminate $\bm{\Phi}_\mathbf{e}^{-1}$. The affine half-space~\eqref{affine_SLS:affine-halfspace} combines both~\eqref{affine_SLS:affine_halfspace_feedback_part} and~\eqref{affine_SLS:affine_halfspace_affine_part}, therefore~\eqref{proof:condition} holds.
\newline Finally, $\mathbf{u} = \tilde{\bm{\Phi}}_\mathbf{u}\tilde{\bm{\Phi}}_\mathbf{x}^{-1}\tilde{\mathbf{x}} = \tilde{\bm{\Phi}}_\mathbf{u}\tilde{\bm{\Phi}}_\mathbf{x}^{-1}\tilde{\bm{\Phi}}_\mathbf{x}\tilde{\bm{\delta}} = \tilde{\bm{\Phi}}_\mathbf{u}\tilde{\bm{\delta}}$ follows trivially.
\end{proof}

\begin{remark}
Theorem~\ref{theorem:affine_SLS} is also valid for LTV systems without any modifications to the proof, by adapting the definitions of $\mathcal{ZA}$ and $\mathcal{ZB}$ accordingly.
\end{remark}

The affine subspace~\eqref{affine_SLS:affine-halfspace} can also be separated and interpreted as the governing equation of two separate dynamical systems using the following corollary.

\begin{corollary}\label{corollary:decoupling}
Given an admissible $(\tilde{\bm{\Phi}}_\mathbf{x}, \tilde{\bm{\Phi}}_\mathbf{u})$ with structure~\eqref{eq:affine_structure}, the affine half-spaces~\eqref{affine_SLS:affine_halfspace_feedback_part} and~\eqref{affine_SLS:affine_halfspace_affine_part} define two decoupled dynamical systems, which evolve according to
\begin{align}
\mathbf{e} &= \mathcal{ZA}\mathbf{e} + \mathcal{ZB}\mathbf{u}_\mathbf{e} + \bm{\delta}, \label{error_system:dynamics}\\
\bm{\phi}_\mathbf{z} &= \mathcal{ZA}\bm{\phi}_\mathbf{z} + \mathcal{ZB}\mathbf{u}_\mathbf{z}, \label{nominal_system:dynamics}
\end{align}
where $\mathbf{e} = \mathbf{x} - \bm{\phi}_\mathbf{z}$. The inputs of the two systems are
${\mathbf{u}_\mathbf{e} = \bm{\Phi}_\mathbf{k}\bm{\Phi}_\mathbf{e}^{-1}\mathbf{e} = \bm{\Phi}_\mathbf{k}\bm{\delta}}$ and $\mathbf{u}_\mathbf{z} = \bm{\phi}_\mathbf{v}$, respectively.
\end{corollary}
\begin{proof}
See Appendix~\ref{apx:proof-corollary}.
\end{proof}

Corollary~\ref{corollary:decoupling} provides the separation into error and nominal dynamics commonly encountered in tube-based robust MPC methods, even though the nominal and error state definitions are not exactly the same. Here, $\bm{\phi}_z$ and $\bm{\Phi}_x^{:,0}x_0$ constitute the nominal state, of which the latter part is contributed by the tube controller, while the nominal state of the standard tube formulation in~\eqref{eq:nominal-dynamics} is independent of the tube controller. The same applies for the nominal input.

Note that~\eqref{affine_SLS:affine_halfspace_affine_part} defines a dynamical system similar to~\eqref{eq:stacked_dynamics}, but with $\bm{\delta} = \bm{0}$ and therefore zero initial state. However, the dynamics~\eqref{affine_SLS:affine_halfspace_affine_part} can similarly be rewritten to allow non-trivial initial states by replacing $\bm{\phi}_\mathbf{z}$ in~\eqref{eq:affine_structure} with $\bm{\phi}_\mathbf{z} - \left( I - \mathcal{ZA}\right)^{-1}\bm{\delta}_\mathbf{z}$, where $\bm{\delta}_\mathbf{z} = \begin{bmatrix} \bm{\phi}_\mathbf{z}^0 & \bm{0}^\top\end{bmatrix}^\top$. All computations in the proof hold for any choice of $\bm{\phi}_\mathbf{z}$, thus~\eqref{affine_SLS:affine_halfspace_affine_part} becomes
\begin{equation}\label{affine_SLS:affine_halfspace_generalized_affine_part}
\begin{bmatrix} I - \mathcal{ZA} & -\mathcal{ZB}\end{bmatrix} \begin{bmatrix} \bm{\phi}_\mathbf{z} \\ \bm{\phi}_\mathbf{v} \end{bmatrix} = \bm{\delta}_\mathbf{z}.
\end{equation}
At the same time, the error state definition in~\eqref{error_system:dynamics} is changed to $\mathbf{e} = \mathbf{x} - \bm{\phi}_\mathbf{z} + \left( I - \mathcal{ZA}\right)^{-1}\!\bm{\delta}_\mathbf{z}$, which removes the contribution of the initial nominal state~$\bm{\phi}_\mathbf{z}^0$ from the error dynamics.

We use Theorem~\ref{theorem:affine_SLS} to formulate the proposed SLTMPC problem, using the diagonally restricted system responses $\bar{\bm{\Phi}}_\mathbf{e}$ and $\bar{\bm{\Phi}}_\mathbf{k}$
\begin{subequations}\label{SLTMPC}
	\begin{alignat}{2}
		\min_{\tilde{\bm{\Phi}}_\mathbf{x},\tilde{\bm{\Phi}}_\mathbf{u}} \quad & \norm{\begin{bmatrix} \mathcal{Q}^\frac{1}{2} & 0 \\ 0 & \mathcal{R}^\frac{1}{2}\end{bmatrix} \begin{bmatrix} \bm{\phi}_\mathbf{z} + \bar{\bm{\Phi}}_\mathbf{e}^{:,0}x_0 \\ \bm{\phi}_\mathbf{v} + \bar{\bm{\Phi}}_\mathbf{k}^{:,0}x_0 \end{bmatrix}}^2_2 \label{SLTMPC:cost}\\
		\textrm{s.t. } & \begin{bmatrix} I - \mathcal{Z}\tilde{\mathcal{A}} & -\mathcal{Z}\tilde{\mathcal{B}}\end{bmatrix} \begin{bmatrix} \tilde{\bm{\Phi}}_\mathbf{x} \\ \tilde{\bm{\Phi}}_\mathbf{u} \end{bmatrix} = I, \\
		& \bm{\phi}_\mathbf{z}^{:N} + \bar{\bm{\Phi}}_\mathbf{e}^{:N,:}\bm{\delta} \in \mathcal{X}^N, &&\hspace{-0.5cm}\forall\mathbf{w} \in \mathcal{W}^{N-1} \label{SLTMPC:state_constraints} \\
        & \bm{\phi}_\mathbf{v} + \bar{\bm{\Phi}}_\mathbf{k}\bm{\delta} \in \mathcal{U}^N, &&\hspace{-0.5cm}\forall\mathbf{w} \in \mathcal{W}^{N-1} \label{SLTMPC:input_constraints} \\
        & \bm{\phi}_\mathbf{z}^{N} + \bar{\bm{\Phi}}_\mathbf{e}^{N,:}\bm{\delta} \in \mathcal{X}_f, &&\hspace{-0.5cm}\forall\mathbf{w} \in \mathcal{W}^{N-1} \label{SLTMPC:terminal_constraint} \\
        & \bm{\delta}^0 = x_k,
	\end{alignat}
\end{subequations}
where $\mathcal{Q}$, $\mathcal{R}$, and $\mathcal{X}_f$ are defined as in~\eqref{MPC:generic}. The input applied to the system is then computed as $u_k = \bm{\phi}_\mathbf{v}^0 + \bm{\Phi}_\mathbf{k}^{0,0}\bm{\delta}^0 = \phi_v^0 + \Phi_u^{0}x_k$. In order to render the constraints~\eqref{SLTMPC:state_constraints} -~\eqref{SLTMPC:terminal_constraint} amenable for optimization, we rely on a standard technique from \df-MPC~\cite{Goulart2006}, which removes the explicit dependency on the disturbance trajectory by employing Lagrangian dual variables. This corresponds to an implicit representation of the error dynamics' reachable sets commonly used for constraint tightening in \tube-MPC.

\begin{remark}
Under the assumption that the sets $\mathcal{X}$, $\mathcal{U}$, and $\mathcal{W}$ are polytopic, the constraint tightening approach using dual variables~\cite{Goulart2006} is equivalent to the standard constraint tightening based on reachable sets, usually employed in \tube-MPC. This can be shown using results from~\cite{Kolmanovsky1998} on support functions, Pontryagin differences, and their relation in the context of polytopic sets. The main difference is that the constraint tightening based on reachable sets is generally performed offline, while the formulation using dual variables is embedded in the online optimization.
\end{remark}

\begin{remark}
The cost~\eqref{SLTMPC:cost} is equivalent to a standard quadratic cost on the nominal variables~\eqref{eq:extended_SLP}, which renders~\eqref{SLTMPC:cost} identical to the quadratic form $\mathbf{z}^T\mathcal{Q}\mathbf{z} + \mathbf{v}^T\mathcal{R}\mathbf{v}$. Other cost functions can be similarly considered, e.g. not only including the nominal states and inputs but also the effect of the tube controller $\mathbf{K}$.
\end{remark}

\begin{remark}
Additional uncertainty in the system matrices~$A,\,B$ can be incorporated in~\eqref{SLTMPC} using an adapted version of the robust\footnote{Please note that the definition of robustness in the context of the SLP does not directly correspond to the robustness definition commonly used in MPC. For more details on the robust SLP formulation see~\cite[Section~2.3]{Anderson2019}.} SLP formulation~\cite[Section~2.3]{Anderson2019}. The robust SLP formulation was already used in~\cite{Chen2020} to derive a tube-based MPC method, which can handle both additive and system uncertainty. The affine extension presented in this paper can be similarly adapted to formulate a SLTMPC method, which can also handle system uncertainties.
\end{remark}

Using Theorem~\ref{theorem:equivalence} we can relate \df-MPC, \tube-MPC, and SLTMPC to each other. While all three methods employ nominal states and nominal inputs, they differ in the definition of the tube controller. Df-MPC does not restrict the tube controller and treats it as a decision variable. Contrary, \tube-MPC fixes the tube controller a~priori and does not optimize over it. SLTMPC positions itself between those two extremes by allowing time-varying controllers, while preventing changes over the prediction horizon. As a result, SLTMPC improves performance over standard \tube-MPC by offering more flexibility, rendering it computationally more expensive, yet cheaper than \df-MPC. The problem sizes of all three methods are stated in Table~\ref{table:complexity}.

\begin{figure*}[t]
\centering\vspace*{-2mm}
\includegraphics[width=0.99\linewidth]{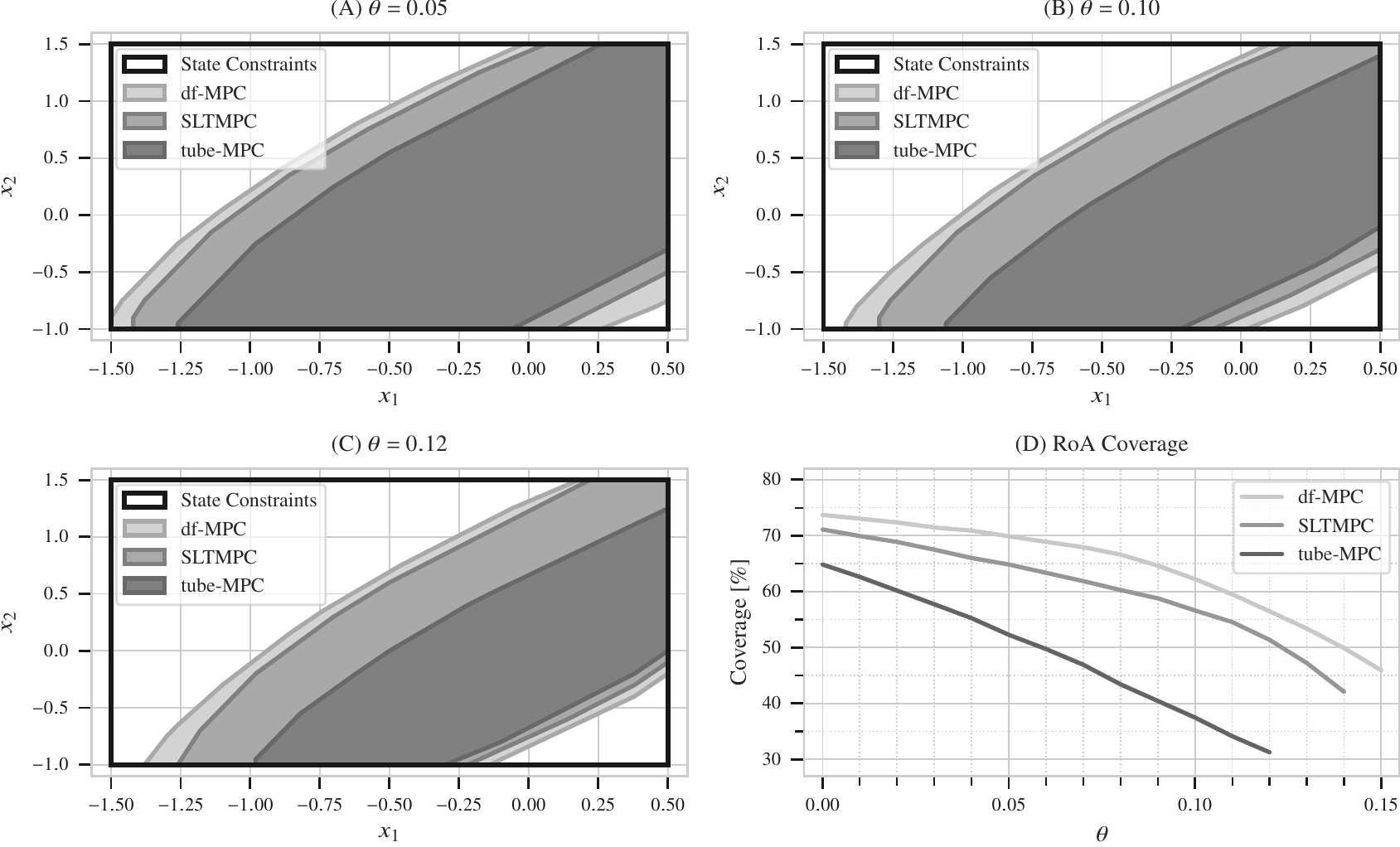}\caption{RoA of SLTMPC and tube-MPC for~${\theta = \{ 0.05, 0.1, 0.12\}}$~(A,B,C) and approximate RoA coverage with respect to the state constraints in percent, as a function of the parameter~$\theta$~(D).}
\label{fig:RoA-coverage}
\end{figure*}

\subsection{Extensions of SLTMPC}
Finally, we outline two extensions of the SLTMPC method to improve scalability and reduce computation time via a distributed and an explicit formulation.

\begin{table}[t]
\centering
\caption{Problem size in terms of the number of variables and constraints for df-MPC, SLTMPC, and tube-MPC.}\label{table:complexity}
\begin{tabular}{@{}lcc@{}}
\toprule
& No. variables & No. constraints \\
\midrule
\df-MPC & $(N+1)(\frac{Nn}{2} + 1)(n+m)$ & $O(N(N+1)n^2)$ \\
SLTMPC  & $2\cdot(N+1)(n+1)(n+m)$ & $O(2(N+1)n^2)$ \\
\tube-MPC  & $(N+1)(n+m)$ & $O((N+1)n^2)$ \\
\bottomrule 
\end{tabular}
\end{table}

\paragraph{Distributed SLTMPC}
Using ideas from distributed control allows us to formulate a distributed version of~\eqref{SLTMPC} similar to the distributed SLP formulation proposed in~\cite{Anderson2019}. Since any structure on $\tilde{\bm{\Phi}}_\mathbf{x}$ and $\tilde{\bm{\Phi}}_\mathbf{u}$ can be enforced directly in the optimization problem~\eqref{SLTMPC}~\cite{Anderson2019}, this provides a straightforward way to formulate a distributed SLTMPC scheme. Using the concept of sparsity invariance~\cite{Furieri2019b}, any distributed structure on $\tilde{\mathbf{K}}$ can be reformulated as structural conditions on $\tilde{\bm{\Phi}}_\mathbf{x}$ and $\tilde{\bm{\Phi}}_\mathbf{u}$. This in turn enables us to rewrite~\eqref{SLTMPC} with the additional constraints $$\tilde{\bm{\Phi}}_\mathbf{x} \in \mathcal{K}_x, \qquad \tilde{\bm{\Phi}}_\mathbf{u} \in \mathcal{K}_u, $$ where $\mathcal{K}_x$ and $\mathcal{K}_u$ are sparsity patterns.

\paragraph{Explicit SLTMPC}
The computational complexity associated with optimization problem~\eqref{SLTMPC} can be significant, especially for long horizons and large state spaces. This motivates an explicit formulation of the SLTMPC method, which allows pre-computation of the control law and the corresponding set of initial states. One way to formulate such an explicit scheme is to solve the Karush-Kuhn-Tucker conditions as a function of the initial state and thus obtain a control policy for each partition of the set of initial states~\cite{Bemporad2000}. In the context of SLS-MPC, an adapted version of this method was proposed in~\cite{Alonso2020}. Another option is to approximate the explicit solution by defining the sets of initial states manually and solve the SLTMPC problem for a whole set of initial states, similar to the method presented in~\cite{Carron2020}.

\section{NUMERICAL RESULTS}\label{sec:numerical_section}
In the following, we show the benefits of the proposed SLTMPC method for an example by comparing it against \df-MPC and \tube-MPC. The example is implemented in Python using CVXPY~\cite{cvxpy} and is solved using MOSEK~\cite{mosek}. The example was run on a machine equipped with an Intel~i9~(\unit[4.3]{GHz}) CPU and \unit[32]{GB} of RAM.

We consider the uncertain LTI system~\eqref{eq:dynamics} with discrete-time dynamic matrices
\begin{equation*}
A = \begin{bmatrix} 1 & 0.15 \\ 0 & 1 \end{bmatrix}, \quad B = \begin{bmatrix} 0.5 \\ 0.5 \end{bmatrix},
\end{equation*}
subject to the polytopic constraints
\begin{equation*}
\begin{bmatrix} -1.5 \\ -1 \end{bmatrix} \!\leq\! \begin{bmatrix} x_1 \\ x_2 \end{bmatrix} \!\leq\! \begin{bmatrix} 0.5 \\ 1.5 \end{bmatrix}, \!\ -1 \!\leq\! u \!\leq\! 1, \!\
\begin{bmatrix} -\theta \\ -0.1 \end{bmatrix} \!\leq\! \begin{bmatrix} w_1 \\ w_2 \end{bmatrix} \!\leq\! \begin{bmatrix} \theta \\ 0.1 \end{bmatrix}
\end{equation*}
with state cost $Q=I$, input cost $R=10$, horizon $N=10$, and parameter $\theta$. As the nominal terminal constraint we choose the origin, i.e. ${z_N = [0,0]^\top}$. We compare the region of attraction (RoA) and performance of SLTMPC~\eqref{SLTMPC} to those of \df-MPC and \tube-MPC, where the constant tube controller~$K$ is computed such that it minimizes the constraint tightening~\cite[Section~7.2]{Limon2010}. Figure~\ref{fig:RoA-coverage}~(A,\,B,\,C) shows the RoAs for three different noise levels, i.e. $\theta = \{ 0.05, 0.1, 0.12\}$. SLTMPC achieves a RoA which is considerably larger than the RoA of \tube-MPC and comparable to the one of \df-MPC. Figure~\ref{fig:RoA-coverage}~(D) further highlights this by approximating the coverage of the RoA, i.e. the area of the state constraints covered by the RoA in percent. RoA coverage of SLTMPC is consistently larger than for \tube-MPC and decreases more slowly as the disturbance parameter $\theta$ is increased. This behavior is comparable to that of \df-MPC, although SLTMPC achieves lower RoA coverage. All three methods become infeasible and thus achieve a coverage of~$0\%$ for $\theta_{tube}=0.13$, $\theta_{SLTMPC}=0.15$, and $\theta_{df} = 0.16$, respectively. Figure~\ref{fig:trajectories} shows the open-loop nominal trajectory and trajectories for 10'000 randomly sampled noise realizations for all three methods. Due to the optimized tube controller and less conservative constraint tightening, \df-MPC and SLTMPC compute trajectories that approach the constraints more closely than \tube-MPC. Table~\ref{table:comp-times} states the computation times and costs of the methods for the noisy trajectories shown in Figure~\ref{fig:trajectories}, highlighting the trade-off in computation time and performance offered by SLTMPC.

\begin{figure}[h!]
\centering\vspace*{-3mm}
\includegraphics[width=0.99\columnwidth]{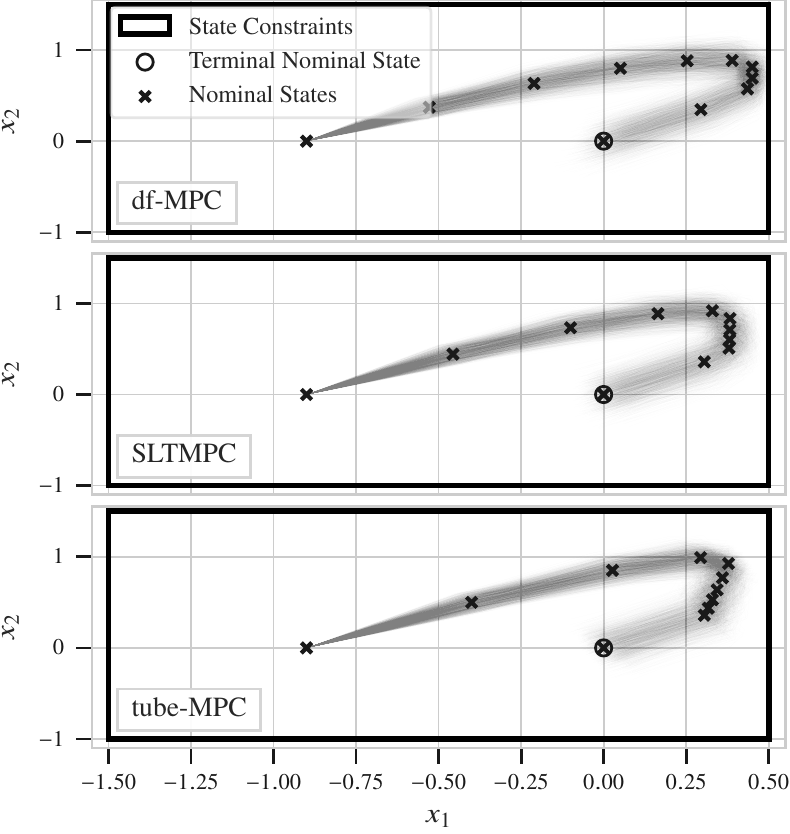}\caption{Nominal state trajectory and 10'000 noise realizations from initial state $x_0 = [-0.9, 0.0]^\top$ and parameter $\theta=0.05$ for df-MPC, SLTMPC, and tube-MPC.}
\label{fig:trajectories}\vspace*{-1mm}
\end{figure}

\section{CONCLUSIONS}\label{sec:conclusions}
This paper has proposed a \tube-MPC method for linear time-varying systems with additive noise based on the system level parameterization (SLP). The formulation was derived by first establishing the equivalence between disturbance feedback MPC (\df-MPC) and the SLP, offering a new perspective on a subclass of robust MPC methods from the angle of SLP. Subsequently, the standard SLP was extended in order to formulate the proposed system level \tube-MPC~(SLTMPC) method. Finally, we showed the effectiveness of the proposed method by comparing it against \df-MPC and \tube-MPC on a numerical example.

\begin{table}[t]
\centering
\caption{Closed-loop costs and computation times for 10'000 noisy trajectories starting in $x_0 = [-0.9, 0.0]^\top$.}\label{table:comp-times}
\begin{tabular}{@{}lcccc@{}}
\toprule
& \multicolumn{2}{c}{Cost [-]} & \multicolumn{2}{c}{Computation Time [ms]} \\ \cmidrule(lr){2-3} \cmidrule(lr){4-5}
& Mean & Std. Deviation & Mean & Std. Deviation \\ \midrule
\df-MPC & 24.61 & 1.76 & 53.21 & 6.66\\
SLTMPC & 26.38 & 1.87 & 33.49 & 6.35\\
\tube-MPC & 30.44 & 1.98 & 6.9 & 1.43\\
\bottomrule 
\end{tabular}
\end{table}

\section*{Acknowledgment}
We would like to thank Carlo Alberto Pascucci from Embotech for the insightful discussions on the topic of this work.

\appendices
\section{Proof of Lemma~\ref{lemma:MPC_to_SLS}}\label{apx:proof-lemma}
\begin{proof}
The matrices in~\eqref{eq:function_of_x0} are defined as
\begin{equation}\label{eq:A_E_def}
\mathbf{A} = \begin{bmatrix} I \\ A \\ A^2 \\ \vdots \\ A^N\end{bmatrix}, \quad \mathbf{E} = \begin{bmatrix} 0 & \dots & \dots & 0\\ I & 0 & \dots & \vdots \\ A  & I & \ddots & \vdots \\ \vdots & \ddots & \ddots & 0 \\ A^{N-1} & \dots & A & I\end{bmatrix},
\end{equation}
and $\mathbf{B}=\begin{bmatrix} \mathbf{A} & \mathbf{E} \end{bmatrix} \mathcal{ZB}$. These quantities can then be used to rewrite~\eqref{SLS:affine-halfspace} as
\begin{align}\label{lemma:second_relation}
\bm{\Phi}_\mathbf{x} &= \left(I - \mathcal{ZA}\right)^{-1} + \left(I - \mathcal{ZA}\right)^{-1}\mathcal{ZB}\bm{\Phi}_\mathbf{u} \nonumber\\
&\mathrel{\overset{\makebox[0pt]{\mbox{\normalfont\tiny\sffamily (i)}}}{=}}
 \begin{bmatrix} \mathbf{A} & \mathbf{E}\, \end{bmatrix} + \begin{bmatrix} \mathbf{A} & \mathbf{E}\, \end{bmatrix}\mathcal{ZB}\bm{\Phi}_\mathbf{u} \nonumber\\
 &= \begin{bmatrix} \mathbf{A} & \mathbf{E}\, \end{bmatrix} + \mathbf{B}\bm{\Phi}_\mathbf{u},
\end{align}
where relation $(i)$ follows from the fact that $\left(I - \mathcal{ZA}\right)^{-1}$ can be represented as a finite Neumann series since $\mathcal{ZA}$ is nilpotent with index $N$, which leads to
\begin{equation*}
\left(I - \mathcal{ZA}\right)^{-1} = \sum_{k=0}^\infty \left(\mathcal{ZA}\right)^k = \sum_{k=0}^N \left(\mathcal{ZA}\right)^k = \begin{bmatrix} \mathbf{A} & \mathbf{E} \end{bmatrix}.
\end{equation*}
Then,~\eqref{lemma:second_relation} proves Lemma~\ref{lemma:MPC_to_SLS}.
\end{proof}

\section{Proof of Corollary~\ref{corollary:decoupling}}\label{apx:proof-corollary}
\begin{proof}
Given an admissible $(\tilde{\bm{\Phi}}_\mathbf{x}, \tilde{\bm{\Phi}}_\mathbf{u})$ with structure~\eqref{eq:affine_structure}, the affine state trajectory is
\begin{equation}
\begin{bmatrix} 1 \\ \mathbf{x} \end{bmatrix} = \mathcal{Z}\tilde{\mathcal{A}}\begin{bmatrix} 1 \\ \mathbf{x} \end{bmatrix} + \mathcal{Z}\tilde{\mathcal{B}}\tilde{\bm{\Phi}}_\mathbf{u}\tilde{\bm{\Phi}}_\mathbf{x}^{-1}\begin{bmatrix} 1 \\ \mathbf{x} \end{bmatrix} + \tilde{\bm{\delta}}, \label{eq:affine_state_trajectory}
\end{equation}
where
\begin{equation}
\tilde{\bm{\Phi}}_\mathbf{u}\tilde{\bm{\Phi}}_\mathbf{x}^{-1} \!=\! \begin{bmatrix} \bm{\phi}_\mathbf{v} & \hspace{-0.5em}\bm{\Phi}_\mathbf{k} \end{bmatrix}\!\begin{bmatrix} 1 & \hspace{-0.5em}\bm{0} \\ -\bm{\Phi}_\mathbf{e}^{-1}\bm{\phi}_\mathbf{z} & \hspace{-0.5em}\bm{\Phi}_\mathbf{e}^{-1} \end{bmatrix} \!=\! \begin{bmatrix} \bm{\phi}_\mathbf{v} \!-\! \bm{\Phi}_\mathbf{k}\bm{\Phi}_\mathbf{e}^{-1}\bm{\phi}_\mathbf{z} \\ \bm{\Phi}_\mathbf{k}\bm{\Phi}_\mathbf{e}^{-1} \end{bmatrix}^T\!\!, \label{eq:affine_SLS_controller}
\end{equation}
then plugging~\eqref{eq:affine_SLS_controller} into~\eqref{eq:affine_state_trajectory} yields
\begin{equation}
\mathbf{x} \!=\! \left(\mathcal{ZA} \!+\! \mathcal{ZB}\bm{\Phi}_\mathbf{k}\bm{\Phi}_\mathbf{e}^{-1}\right)\mathbf{x} \!+\! \mathcal{ZB}\bm{\phi}_\mathbf{v} \!+\! \bm{\delta} \!-\! \mathcal{ZB}\bm{\Phi}_\mathbf{k}\bm{\Phi}_\mathbf{e}^{-1}\bm{\phi}_\mathbf{z}. \label{eq:affine_evolution}
\end{equation}
Adding $-\left(I - \mathcal{ZA}\right)\bm{\phi}_\mathbf{z}$ on both sides of~\eqref{eq:affine_evolution}, results in
\begin{align}
&\left(I - \mathcal{ZA} - \mathcal{ZB}\bm{\Phi}_\mathbf{k}\bm{\Phi}_\mathbf{e}^{-1}\right)\left(\mathbf{x} - \bm{\phi}_\mathbf{z} \right) = \nonumber\\ &- \left(I - \mathcal{ZA}\right)\bm{\phi}_\mathbf{z} + \mathcal{ZB}\bm{\phi}_\mathbf{v} + \bm{\delta}. \label{eq:condition_for_lemma}
\end{align}
Using~\eqref{affine_SLS:affine_halfspace_affine_part}, we conclude that
${-\left(I - \mathcal{ZA}\right)\bm{\phi}_\mathbf{z} + \mathcal{ZB}\bm{\phi}_\mathbf{v} = \bm{0}},$
hence~\eqref{eq:condition_for_lemma} becomes $\left(I - \mathcal{ZA} - \mathcal{ZB}\bm{\Phi}_\mathbf{k}\bm{\Phi}_\mathbf{e}^{-1}\right)\mathbf{e} = \bm{\delta},$ and thus
\begin{equation}
\mathbf{e} = \left(I - \mathcal{ZA} - \mathcal{ZB}\bm{\Phi}_\mathbf{k}\bm{\Phi}_\mathbf{e}^{-1}\right)^{-1}\bm{\delta} = \bm{\Phi}_\mathbf{e}\bm{\delta},
\end{equation}
which describes the closed loop behavior of a system with state $\mathbf{e} = \mathbf{x} - \bm{\phi}_\mathbf{z}$ under the control law $\mathbf{u}_\mathbf{e} = \bm{\Phi}_\mathbf{k}\bm{\Phi}_\mathbf{e}^{-1}\mathbf{e}$. The dynamics~\eqref{nominal_system:dynamics} are directly given by~\eqref{affine_SLS:affine_halfspace_affine_part}. Therefore,~\eqref{error_system:dynamics} and~\eqref{nominal_system:dynamics} define decoupled dynamical systems, with the inputs $\mathbf{u}_\mathbf{e} = \bm{\Phi}_\mathbf{k}\bm{\delta}$ and $\mathbf{u}_\mathbf{z} = \bm{\phi}_\mathbf{v}$.
\end{proof}






\bibliography{bibliography.bib}

\begin{thebibliography}{10}
\providecommand{\url}[1]{#1}
\csname url@rmstyle\endcsname
\providecommand{\newblock}{\relax}
\providecommand{\bibinfo}[2]{#2}
\providecommand\BIBentrySTDinterwordspacing{\spaceskip=0pt\relax}
\providecommand\BIBentryALTinterwordstretchfactor{4}
\providecommand\BIBentryALTinterwordspacing{\spaceskip=\fontdimen2\font plus
\BIBentryALTinterwordstretchfactor\fontdimen3\font minus
  \fontdimen4\font\relax}
\providecommand\BIBforeignlanguage[2]{{%
\expandafter\ifx\csname l@#1\endcsname\relax
\typeout{** WARNING: IEEEtran.bst: No hyphenation pattern has been}%
\typeout{** loaded for the language `#1'. Using the pattern for}%
\typeout{** the default language instead.}%
\else
\language=\csname l@#1\endcsname
\fi
#2}}

\bibitem{Bemporad1999}
A.~Bemporad and M.~Morari, ``Robust model predictive control: A survey,'' in
  \emph{Robustness in identification and control}.\hskip 1em plus 0.5em minus
  0.4em\relax Springer, 1999, pp. 207--226.

\bibitem{Lofberg2003}
J.~L{\"{o}}fberg, ``{Approximations of closed-loop minimax MPC},'' in
  \emph{Proc. 42nd IEEE Conf. Decis. Control}, vol.~2, 2003, pp. 1438--1442.

\bibitem{Goulart2006}
P.~J. Goulart, E.~C. Kerrigan, and J.~M. Maciejowski, ``{Optimization over
  state feedback policies for robust control with constraints},''
  \emph{Automatica}, vol.~42, no.~4, pp. 523--533, 2006.

\bibitem{Chisci2001}
L.~Chisci, J.~A. Rossiter, and G.~Zappa, ``Systems with persistent
  disturbances: predictive control with restricted constraints,''
  \emph{Automatica}, vol.~37, no.~7, pp. 1019--1028, 2001.

\bibitem{Langson2004}
W.~Langson, I.~Chryssochoos, S.~Rakovi{\'c}, and D.~Q. Mayne, ``Robust model
  predictive control using tubes,'' \emph{Automatica}, vol.~40, no.~1, pp.
  125--133, 2004.

\bibitem{Mayne2005}
D.~Q. Mayne, M.~M. Seron, and S.~Rakovi{\'c}, ``Robust model predictive control
  of constrained linear systems with bounded disturbances,'' \emph{Automatica},
  vol.~41, no.~2, pp. 219--224, 2005.

\bibitem{Anderson2019}
J.~Anderson, J.~C. Doyle, S.~H. Low, and N.~Matni, ``{System level
  synthesis},'' \emph{Annu. Rev. Control}, vol.~47, pp. 364--393, 2019.

\bibitem{Wang2018a}
Y.~S. Wang, N.~Matni, and J.~C. Doyle, ``{Separable and Localized System-Level
  Synthesis for Large-Scale Systems},'' \emph{IEEE Trans. Automat. Contr.},
  vol.~63, no.~12, pp. 4234--4249, 2018.

\bibitem{Kogel2020}
M.~K\"ogel and R.~Findeisen, ``{Robust MPC with Reduced Conservatism Blending
  Multiples Tubes},'' in \emph{Proc. of American Control Conference (ACC)},
  2020, pp. 1949--1954.

\bibitem{Chen2020}
S.~Chen, H.~Wang, M.~Morari, V.~M. Preciado, and N.~Matni, ``{Robust
  Closed-loop Model Predictive Control via System Level Synthesis},'' in
  \emph{Proc. 59th IEEE Conf. Decis. Control}, 2020, pp. 2152--2159.

\bibitem{Dean2018}
S.~{Dean}, S.~{Tu}, N.~{Matni}, and B.~{Recht}, ``Safely learning to control
  the constrained linear quadratic regulator,'' in \emph{2019 American Control
  Conference (ACC)}, 2019, pp. 5582--5588.

\bibitem{Alonso2019}
C.~A. {Alonso} and N.~{Matni}, ``{Distributed and Localized Closed Loop Model
  Predictive Control via System Level Synthesis},'' in \emph{Proc. 59th IEEE
  Conf. Decis. Control}, 2020, pp. 5598--5605.

\bibitem{Alonso2020}
C.~A. {Alonso}, N.~{Matni}, and J.~{Anderson}, ``{Explicit Distributed and
  Localized Model Predictive Control via System Level Synthesis},'' in
  \emph{Proc. 59th IEEE Conf. Decis. Control}, 2020, pp. 5606--5613.

\bibitem{Li2020}
\BIBentryALTinterwordspacing
J.~S. Li, C.~A. Alonso, and J.~C. Doyle, ``{Frontiers in Scalable Distributed
  Control: SLS, MPC, and Beyond},'' 2020. [Online]. Available:
  \url{https://arxiv.org/abs/2010.01292}
\BIBentrySTDinterwordspacing

\bibitem{Rawlings2009}
J.~B. Rawlings, D.~Q. Mayne, and M.~Diehl, \emph{{Model Predictive Control:
  Theory, Computation, and Design}}.\hskip 1em plus 0.5em minus 0.4em\relax Nob
  Hill Publishing, 2017.

\bibitem{Rakovic2012}
S.~V. {Rakovic}, B.~{Kouvaritakis}, M.~{Cannon}, C.~{Panos}, and
  R.~{Findeisen}, ``{Parameterized Tube Model Predictive Control},'' \emph{IEEE
  Trans. Automat. Contr.}, vol.~57, no.~11, pp. 2746--2761, 2012.

\bibitem{Tseng2020a}
\BIBentryALTinterwordspacing
S.-H. Tseng and J.~S. Li, ``{SLSpy: Python-Based System-Level Controller
  Synthesis Framework},'' 2020. [Online]. Available:
  \url{https://arxiv.org/abs/2004.12565}
\BIBentrySTDinterwordspacing

\bibitem{Tseng2020b}
S.~H. {Tseng}, C.~A. {Alonso}, and S.~{Han}, ``{System Level Synthesis via
  Dynamic Programming},'' in \emph{Proc. 59th IEEE Conf. Decis. Control}, 2020,
  pp. 1718--1725.

\bibitem{Stewart1998}
G.~W. Stewart, \emph{Matrix Algorithms: Volume 1: Basic Decompositions}.\hskip
  1em plus 0.5em minus 0.4em\relax SIAM, 1998.

\bibitem{Kolmanovsky1998}
I.~Kolmanovsky and E.~G. Gilbert, ``Theory and computation of disturbance
  invariant sets for discrete-time linear systems,'' \emph{Mathematical
  Problems in Engineering}, vol.~4, 1998.

\bibitem{Furieri2019b}
L.~{Furieri}, Y.~{Zheng}, A.~{Papachristodoulou}, and M.~{Kamgarpour},
  ``{Sparsity Invariance for Convex Design of Distributed Controllers},''
  \emph{IEEE Transactions on Control of Network Systems}, vol.~7, no.~4, pp.
  1836--1847, 2020.

\bibitem{Bemporad2000}
A.~Bemporad, M.~Morari, V.~Dua, and E.~N. Pistikopoulos, ``The explicit
  solution of model predictive control via multiparametric quadratic
  programming,'' in \emph{2000 American Control Conference (ACC)},
  vol.~2.\hskip 1em plus 0.5em minus 0.4em\relax IEEE, 2000, pp. 872--876.

\bibitem{Carron2020}
\BIBentryALTinterwordspacing
A.~Carron, J.~Sieber, and M.~N. Zeilinger, ``{Distributed Safe Learning using
  an Invariance-based Safety Framework},'' 2020. [Online]. Available:
  \url{https://arxiv.org/abs/2007.00681}
\BIBentrySTDinterwordspacing

\bibitem{cvxpy}
S.~Diamond and S.~Boyd, ``{CVXPY: A Python-Embedded Modeling Language for
  Convex Optimization},'' \emph{Journal of Machine Learning Research}, vol.~17,
  no.~83, pp. 1--5, 2016.

\bibitem{mosek}
\BIBentryALTinterwordspacing
{MOSEK ApS}, \emph{MOSEK Optimizer API for Python. Version 9.1.}, 2019.
  [Online]. Available: \url{https://docs.mosek.com/9.1/pythonapi/index.html}
\BIBentrySTDinterwordspacing

\bibitem{Limon2010}
D.~Lim{\'o}n, I.~Alvarado, T.~Alamo, and E.~F. Camacho, ``Robust tube-based
  {MPC} for tracking of constrained linear systems with additive
  disturbances,'' \emph{Journal of Process Control}, vol.~20, no.~3, pp.
  248--260, 2010.

\end{thebibliography}

\end{document}